\documentclass[a4paper,envcountsame]{llncs}

\usepackage{geometry}
\let\llncssubparagraph\subparagraph
\let\subparagraph\paragraph
\usepackage[compact]{titlesec}
\let\subparagraph\llncssubparagraph

\usepackage{enumitem}
\setlist[description]{leftmargin=\parindent,labelindent=0mm}

\usepackage{amsmath}
\usepackage{amssymb}
\usepackage{xspace}
\usepackage{verbatim}
\usepackage{bbm}

\usepackage{mathtools}
\usepackage{cite}
\usepackage{xspace}

\usepackage[disable]{todonotes}

\usepackage{caption}
\usepackage{subcaption}
\newcommand{\ph}{\ensuremath{\phi}\xspace}
\newcommand{\eps}{\ensuremath{\epsilon}\xspace}

\newcommand{\PPAD}{\ensuremath{\mathtt{PPAD}}\xspace}
\newcommand{\NP}{\ensuremath{\mathtt{NP}}\xspace}

\newcommand{\ETR}{\ensuremath{\mathtt{ETR}}\xspace}

\newcommand{\xcal}{\ensuremath{\mathcal{X}}\xspace}

\newcommand{\xbf}{\ensuremath{\mathbf{x}}\xspace}
\newcommand{\sbf}{\ensuremath{\mathbf{s}}\xspace}
\newcommand{\sbfne}{\ensuremath{\mathbf{s}_\text{NE}}\xspace}
\newcommand{\tbf}{\ensuremath{\mathbf{t}^\text{NE}}\xspace}
\newcommand{\pbf}{\ensuremath{\mathbf{p}}\xspace}

\newcommand{\payv}{\ensuremath{\pbf}\xspace}

\newcommand{\supp}{\mathrm{supp}}

\newcommand{\tilg}{\ensuremath{\widetilde{G}}\xspace}

\newcounter{probc}
\newcommand{\probc}[1]{\refstepcounter{probc}\label{#1}}

\DeclareMathOperator{\addf}{add}
\DeclareMathOperator{\mergef}{merge}

\def\dosth#1{\ifx###1##\else\dofirst#1\anytoken\fi}
\def\doagain#1\anytoken{\dosth{#1}}
\def\payoffpairs#1#2#3{\m=#1\multiply\m by 4 \advance\m by -1 \n=1
  \def\dofirst##1{\put(\n,-\m){\makebox(0,0){\strut##1}}\advance\n by 4 \doagain}%
  \dosth{#2\strut}%
  \m=#1\multiply\m by 4 \advance\m by -3 \n=3 \dosth{#3\strut}}
\def\singlepayoffs#1#2{\m=#1\multiply\m by 4 \advance\m by -2 \n=2
  \def\dofirst##1{\put(\n,-\m){\makebox(0,0){\strut##1}}\advance\n by 4 \doagain}%
  {\large\dosth{#2\strut}}}
\newcommand{\bimatrixgame}[8]{%
\setlength{\unitlength}{#1}%
\newcount\rows
\newcount\cols
\rows=#2
\cols=#3
\newcount\rowcoord
\newcount\colcoord
\rowcoord=\rows
\colcoord=\cols
\multiply\rowcoord by 4
\multiply\colcoord by 4
\newcount\m
\newcount\n
\m=\rowcoord
\n=\colcoord
\advance\m by 2 
\advance\n by 2 
\begin{picture}(\n,\m)(-2,-\rowcoord)
\m=\rows
\n=\cols
\advance\m by 1
\advance\n by 1 
\thinlines
\multiput(0,0)(0,-4){\m}{\line(1,0){\colcoord}}
\multiput(0,0)(4,0){\n}{\line(0,-1){\rowcoord}}
\put(0,0){\line(-1,1){2}}
\put(-1.5,0.5){\makebox(0,0)[r]{#4}}  
\put(-.7,2.0){\makebox(0,0)[l]{#5}}   
\n=2
\def\dofirst##1{\put(-0.8,-\n){\makebox(0,0)[r]{\strut##1}}\advance\n by 4
   \doagain}
\dosth{#6\strut} 
\n=2
\def\dofirst##1{\put(\n,1.0){\makebox(0,0){\strut##1}}\advance\n by 4
   \doagain}
\dosth{#7\strut}#8%
\end{picture}}
\begin{document}

\title{Computing Constrained Approximate Equilibria in Polymatrix Games}

\author{
Argyrios Deligkas\inst{1} \and
John Fearnley\inst{2} \and
Rahul Savani\inst{2}
}

\institute{Technion, Israel \and University of Liverpool, UK}

\maketitle
\begin{abstract}


	This paper is about computing \emph{constrained} approximate Nash equilibria in polymatrix games,
	which are succinctly represented many-player games defined by an interaction
	graph between the players.  In a recent breakthrough, Rubinstein showed that
	there exists a small constant~$\eps$, such that it is \PPAD-complete to find
	an (unconstrained) $\eps$-Nash equilibrium of a polymatrix game.  
	In the first part of the paper, we show that is \NP-hard to decide if a 
	polymatrix game has a constrained approximate equilibrium for 9 natural constraints
	and \emph{any} non-trivial approximation guarantee.
	These results hold even for planar bipartite polymatrix games with degree
	3 and at most 7 strategies per player, and all non-trivial approximation
	guarantees.  
	These results
	stand in contrast to similar results for bimatrix games, which obviously need
	 a non-constant number of actions, and which rely on
	stronger complexity-theoretic conjectures such as the exponential time
	hypothesis. 
	In the second part, we provide 
	a deterministic QPTAS for interaction graphs with bounded treewidth and with
	logarithmically many actions per player that can compute constrained approximate
	equilibria for a wide family of constraints that cover many of the constraints dealt with
	in the first part.

\end{abstract}

\section{Introduction}
\label{sec:intro}

In this paper we study \emph{polymatrix games}, which provide a succinct
representation of a many-player game. In these games, each
player is a vertex in a graph, and each edge of the graph is a bimatrix game.
Every player chooses a single
strategy and plays it in \emph{all} of the bimatrix games that he is
involved in, and his payoff is the \emph{sum} of the payoffs that he obtains
from each individual edge game. 



A fundamental problem in algorithmic game theory is to design efficient
algorithms for computing \emph{Nash equilibria}. Unfortunately, even in 
bimatrix games, this is~\PPAD-complete~\cite{CDT09,DGP09}, 
which probably rules out efficient algorithms.
Thus, attention has shifted to
\emph{approximate} equilibria. 
There are two natural notions of an approximate equilibrium. An \emph{$\epsilon$-Nash equilibrium}
($\epsilon$-NE) requires that each player has an expected payoff that is within
$\epsilon$ of their best response payoff. An \emph{$\epsilon$-well-supported
Nash equilibrium} ($\epsilon$-WSNE) requires that all players only play pure
strategies whose payoff is within $\epsilon$ of the best response payoff. 

\paragraph{\bf Constrained approximate equilibria.}

Sometimes, it is not enough to find an approximate NE,
but instead we want to find one that satisfies
certain constraints, such as having high social welfare. For bimatrix games, the
algorithm of 
Lipton, Markakis, and Mehta (henceforth LMM)
can be adapted to provide a quasi-polynomial time
approximation scheme (QPTAS) for this task~\cite{LMM03}: we can find
in $m^{O(\frac{\ln m}{\epsilon^2})}$ time
 an $\eps$-NE whose
social welfare is at least as good as any $\eps'$-NE where $\eps'<\eps$.

A sequence of papers~\cite{HK11,ABC13,BKW15,DFS16} has shown
that polynomial time algorithms for finding $\eps$-NEs with good social welfare
are unlikely to exist, subject to various hardness assumptions such as ETH.
These hardness results carry over to a range of
other properties, and apply for all $\eps < \frac{1}{8}$~\cite{DFS16}.

\paragraph{\bf Our contribution.} 
We show that
deciding whether there is an $\eps$-NE with good social welfare in a polymatrix
game is \NP-complete for all $\eps \in [0, 1]$. We then study a variety of
further constraints (Table~\ref{tab:problems}). For each
one, we show that deciding whether there is an $\eps$-WSNE that satisfies
the constraint is \NP-complete for all $\eps \in (0,1)$.
Our results hold even when the game is a planar
bipartite graph with degree at most~3, where each player has at most 7
actions.

To put these results into context, let us contrast them with the known
lower bounds for bimatrix games, which also apply directly to polymatrix games.
Those
results~\cite{HK11,ABC13,BKW15,DFS16} imply that one cannot hope to find an
algorithm that is better than a QPTAS for polymatrix games when $\eps <
\frac{1}{8}$. In comparison, our results show a stronger \NP-hardness result,
apply to all  $\eps$ in the range $(0, 1)$, and hold
even when the players have constantly many actions.

We then study the problem of computing constrained approximate equilibria in
polymatrix games with restricted graphs. Although our hardness results
apply to a broad class of graphs, 
bounded treewidth graphs do not fall within their
scope. A recent result of Ortiz and Irfan~\cite{OI16} has provided a
QPTAS for finding $\eps$-NEs in polymatrix games with bounded treewidth where
every player has at most logarithmically many actions.
We devise a dynamic programming algorithm for finding approximate equilibria in
polymatrix games with bounded treewidth. 
Much like the algorithm in~\cite{OI16}, we discretize both the strategy and payoff
spaces, and obtain a complexity result that matches theirs. However, our
algorithm works directly on the game, avoiding the reduction to a CSP used in their result.

The main benefit is that this algorithm can be adapted to provide a QPTAS
for constrained
approximate Nash equilibria. We introduce \emph{one variable
decomposable (OVD)} constraints, which are a broad class of optimization constraints,
including many of the problems listed in Table~\ref{tab:problems}. We show
that our algorithm can be adapted to find good approximate equilibria relative
to an OVD constraint. Initially, we do this for the restricted class of
\emph{$k$-uniform} strategies: we can find a $k$-uniform $1.5\eps$-NE whose
value is better than any $k$-uniform $\eps/4$-NE. Note
that this is similar to the guarantee given by the LMM technique in bimatrix
games.
We extend this beyond the class of
$k$-uniform strategies for constraints that are defined by a linear combination
of the payoffs, such as social welfare. In this case, we find a $1.5\eps$-NE
whose value is within $O(\eps)$ of \emph{any} $\eps/8$-NE.



\paragraph{\bf Related work.}
%

Barman, Ligett and Piliouras~\cite{BLP15} have provided a randomised QPTAS for
polymatrix games played on trees. Their algorithm is also a dynamic programming
algorithm that discretizes the strategy space using the notion of a $k$-uniform
strategy. Their algorithm is a QPTAS for general polymatrix games on trees and
when the number of pure strategies for every player is bounded by a constant
they get an expected polynomial-time algorithm (EPTAS). 

The work of Ortiz and Irfan~\cite{OI16} applies to a much wider class of games
that they call graphical multi-hypermatrix games. They provide a QPTAS for the
case where the interaction hypergraph has bounded hypertreewidth. This class
includes polymatrix games that have bounded treewdith and logarithmically many
actions per player. For the special cases of tree polymatrix games and tree
graphical games they go further and provide explicit dynamic programming
algorithms that work directly on the game, and avoid the need to solve a
CSP.

Gilboa and Zemel~\cite{GZ89} showed that it is \NP-complete to decide whether there exist
Nash equilibria in bimatrix games with some properties, such as high
social welfare. Conitzer and Sandholm~\cite{CS08} extended the list of
\NP-complete problems of~\cite{GZ89}. Recently, Garg et al.~\cite{GMVY} and Bilo and
Mavronicolas~\cite{BM16,BM17} extended these results to many-player games and
provided analogous \ETR-completeness results.


Computing approximate equilibria in
bimatrix games has been well
studied~\cite{BBM10,CDFFJS,DMP07,DMP09,FGSS12,KS10,TS08}, but there has been
less work for polymatrix games~\cite{polymatrix,DFIS,BL15}.
Rubinstein~\cite{R15} has shown that there is a small
constant $\eps$ such that finding an $\eps$-NE of a polymatrix game is
\PPAD-complete. 
For constrained \eps-NE, the only positive results were for
bimatrix games and gave algorithms
for finding $\eps$-NE with constraints on payoffs~\cite{CFJ15, CFJ16}.

\section{Preliminaries}

We start by fixing some notation. We use $[k]$ to denote the set of integers
$\{1, 2, \ldots, k\}$, and when a universe $[k]$ is clear, we will use $\bar{S}
= \{ i \in [k], i \notin S\}$ to denote the complement of $S \subseteq [k]$.
For a $k$-dimensional vector $x$, we use $x_{-S}$ to denote the elements of $x$
with indices~$\bar{S}$, and in the case where $S = \{i\}$ has only one
element, we simply write $x_{-i}$ for $x_{-S}$.

An $n$-player polymatrix game is defined by an
undirected graph $G=(V, E)$ with $n$ vertices, where each vertex is
a player. The edges of the graph specify which players interact with each other.
For each $i \in [n]$, we use $N(i) = \{j \; : \; (i,j) \in E\}$ to
denote the neighbors of player $i$.
Each edge $(i, j) \in E$ specifies a bimatrix game to be played between
players $i$ and $j$. Each player $i \in [n]$ has a fixed number of pure
strategies $m$, so the bimatrix game on edge $(i, j) \in E$ is
specified by an $m \times m$ matrix $A_{ij}$, which gives the payoffs for player
$i$, and an $m \times m$ matrix $A_{ji}$, which gives the payoffs for player
$j$. We allow the individual payoffs in each matrix to be an arbitrary rational number. We make the standard normalisation assumption that
the maximum payoff each player can obtain under any strategy profile is 1 and
the minimum is zero, unless specified otherwise. This can be achieved for
example by using the procedure described in~\cite{polymatrix}.
A \emph{subgame} of a polymatrix game is obtained by restricting ignoring edges that
are not contained within a given subgraph of the game's interaction graph~$G$.

A \emph{mixed strategy} for player $i$ is a probability distribution over player 
$i$'s pure strategies. 
A \emph{strategy profile} specifies a mixed strategy for every player.
Given a 
strategy profile $\sbf = (s_1,\ldots,s_n)$, the pure 
strategy payoffs, or the payoff vector, of player $i$ under $\sbf$, where only 
$\sbf_{-i}$ is relevant, is the sum of the pure strategy payoffs that he obtains
in each of the bimatrix games that he plays. Formally, we define: 
%
$\pbf_i(\sbf) := \sum_{j \in N(i)}A_{ij}s_j$.
The \emph{expected} payoff of player $i$ under the strategy profile \sbf is 
defined as $s_i \cdot \pbf_i(\sbf)$. 
The \emph{regret} of player $i$ under \sbf the is difference between $i$'s best
response payoff against $\sbf_{-i}$ and between $i$'s payoff under \sbf.
If a strategy has regret \eps, we say that the strategy is an \eps-best response.
A strategy profile~\sbf is an \eps-Nash 
equilibrium, or \eps-NE, if no player can increase his utility more than 
\eps by unilaterally from \sbf, i.e., if the regret of every player is at most \eps.
Formally, \sbf is an \eps-NE if for every player
 $i \in [n]$ it holds that $s_i \cdot \pbf_i(\sbf) \geq \max \pbf_i(\sbf) - \eps$.
A strategy profile~\sbf is an \eps-well-supported Nash 
equilibrium, or \eps-WSNE, if 
if the regret of every pure strategy played with positive probability is at most \eps.
Formally, \sbf is an \eps-WSNE if for every player~$i \in [n]$ 
it holds that for all $j \in \supp(s_i) = \{k \in [m] \ |\ (s_i)_k > 0\}$ we 
have $(\pbf_i(\sbf))_j \geq \max \pbf_i(\sbf) - \eps$.




\section{Decision problems for approximate equilibria}

In this section, we show \NP-completeness for nine decision problems related to 
constrained approximate Nash equilibria in polymatrix games. 
Table~\ref{tab:problems} contains the list of the problems we study\footnote{
Given probability distributions $\xbf$ and $\xbf'$, the \emph{TV} distance
between them is $\max_i \{ |\xbf_i - \xbf'_i| \}$. The TV distance between strategy profiles 
$\sbf=(s_1, \ldots, s_n)$ and 
$\sbf'=(s'_1, \ldots, s'_n)$ is 
the maximum TV distance of $s_i$ and $s'_i$ over all $i$.}.
For Problem 1, we show hardness for all $\eps \in [0,1]$. For 
the remaining problems, we show hardness for all $\eps \in (0,1)$, i.e., for all approximate equilibria 
except exact equilibria ($\eps=0$), and trivial approximations ($\eps=1$).
All of these problems are contained in \NP because
a ``Yes'' instance
can be witnessed by a suitable approximate equilibrium (or two in the case 
of Problem~\ref{probc:difcomp}).
The starting point for all of our hardness results is 
 the \NP-complete problem Monotone 1-in-3 SAT.

\probc{probc:largetp}
\probc{probc:smalltp}
\probc{probc:smallp}
\probc{probc:restrictedweak}
\probc{probc:difcomp}
\probc{probc:maxprob}
\probc{probc:totalmaxsupport}
\probc{probc:minmaxsupport}
\probc{probc:maxsupport}

\renewcommand*{\arraystretch}{1.6}
\setlength{\tabcolsep}{7pt}
\begin{table}[tpb!]
\centering
\resizebox{\textwidth}{!}{
\begin{tabular}{|p{0.45\textwidth}|p{0.55\textwidth}|}
\hline
\textbf{Problem description} & \textbf{Problem definition} \\
\hline
\hline
Problem~\ref{probc:largetp}: Large total payoff $u \in (0,n]$
&
Is there an \eps-NE \sbf such that $\sum_{i \in [n]} \pbf_i(\sbf) \geq u$?\\
\hline
Problem~\ref{probc:smalltp}: Small total payoff $u \in [0,n)$
&
Is there an \eps-WSNE \sbf such that $\sum_{i \in [n]} \pbf_i(\sbf) \leq u$? \\
\hline
Problem~\ref{probc:smallp}: Small payoff $u \in [0,1)$
&
Is there an \eps-WSNE \sbf such that  $\min_i \pbf_i(\sbf) \leq u$?\\
\hline
Problem~\ref{probc:restrictedweak}: Restricted support $S \subset [n]$
&
Is there an \eps-WSNE \sbf with $\supp(s_1) \subseteq S$? \\
\hline
Problem~\ref{probc:difcomp}: 
Two \eps-WSNE $d \in (0,1]$ apart in Total Variation
 (TV) distance
&
Are there two \eps-WSNE with TV distance $\ge d$? \\
\hline
Problem~\ref{probc:maxprob}: Small largest probability $p \in (0,1)$
&
Is there an \eps-WSNE \sbf with $\max_j s_1(j) \leq p$?\\
\hline

Problem~\ref{probc:totalmaxsupport}: Large total support size $k \in [n \cdot m]$
&
Is there an \eps-WSNE \sbf such that $\sum_{i \in [n]} |\supp(s_i)| \geq k$?\\
\hline
Problem~\ref{probc:minmaxsupport}: Large smallest support size $k \in [n]$
&
Is there an \eps-WSNE \sbf such that $\min_i |\supp(s_i)| \geq k$?\\
\hline
Problem~\ref{probc:maxsupport}: Large support size $k \in [n]$ 
&
Is there an \eps-WSNE \sbf such that $|\supp(s_1)| \geq k$?\\
\hline
\end{tabular}
}
\medskip
\caption{
The decision problems that we consider. 
All problems take as input an $n$-player polymatrix game with $m$ actions for 
each player and an $\eps \in [0,1]$.
}
\label{tab:problems} 
\end{table}

\begin{definition}[Monotone 1-in-3 SAT]
Given a monotone boolean CNF formula $\phi$ with exactly 3 distinct variables per clause,
decide if there exists a satisfying assignment in which exactly
one variable in each clause is true. We call such an assignment a 1-in-3 satisfying assignment.
\end{definition}	

\noindent
Every formula~$\phi$, with
variables $V = \{x_1, \ldots, x_n\}$ and clauses $C = \{y_1, \dots, y_m\}$,
can be represented as a bipartite graph between $V$ and $C$, with 
an edge between $x_i$ and
$y_j$ if and only if $x_i$ appears in clause $y_j$. We assume,
without loss of generality, that this graph is connected. 
We say that $\phi$ is \emph{planar} if the corresponding graph is planar.
Recall that a graph is called~\emph{cubic} if the degree of every vertex
is exactly three.
We use the following result of Moore and Robson~\cite{MR01}.

\begin{theorem}[Section 3.1~\cite{MR01}]
\label{thm:mr01}
Monotone 1-in-3 SAT is \NP-complete even when the formula corresponds to a planar cubic graph.
\end{theorem}
From now on, we assume that $\phi$ is 
a monotone planar cubic formula. 
We say that \ph is a ``Yes'' instance if $\phi$ admits a 1-in-3 satisfying assignment.

\paragraph{\bf Large total payoff for $\eps$-NEs.}
\label{sec:ltp}

Problem~\ref{probc:largetp} asks whether there
exists an $\eps$-NE with high social
welfare. We show that this is \NP-complete for every
$\eps \in [0,1]$, even when the interaction graph for the polymatrix game is
planar, bipartite, cubic, and each player has at most
three pure strategies. 

\paragraph{\bf Construction.}
Given a formula $\phi$, we construct a polymatrix game $G$ 
with $m+n$ players as follows. 
For each variable $x_i$ we create a player $v_i$ and for each clause $y_j$ we 
create a player $c_j$. 
We now use $V$ to denote the set of variable players and $C$ to denote the clause players.
The interaction graph is the bipartite graph between $V$ and $C$ described
above.
%
Each edge game has the same structure. 
Every player in~$V$ has two pure strategies called True and False, while
every player in $C$ has three pure strategies that depend on the three variables
in the clause.
If clause $y_j$ contains variables~$x_i, x_k, x_l$, then 
player $c_j$ has pure strategies $i, k$ and $l$.
The game played between $v_i$ and $c_j$ is shown on the left in
Figure~\ref{fig:combined}.
The bimatrix games for $v_k$ and~$v_l$ are defined analogously.

\begin{figure}[tpb!]
\label{fig:fig1a2}
\begin{subfigure}{0.5\textwidth}
\begin{center}
\bimatrixgame{3.5mm}{3}{2}{Player $c_j$}{Player $v_i$}%
{{$i$}{$k$}{$l$}}%
{{True}{False}}%
{
\payoffpairs{1}{{1}{0}}{{0}{-1}}
\payoffpairs{2}{{$-1$}{$0$}}{{0}{$0$}}
\payoffpairs{3}{{$-1$}{$0$}}{{0}{$0$}}
}
\end{center}
\label{fig:base}
\end{subfigure}
%
\begin{subfigure}{0.5\textwidth}
\begin{center}
\resizebox{0.8\textwidth}{!}{
\bimatrixgame{3.75mm}{4}{3}{Player $c_j$}{Player $v_i$}%
{{$i$}{$k$}{$l$}{Out}}%
{{True}{False}{Out}}%
{
\payoffpairs{1}{{$\kappa$}{$c\cdot\kappa$}{$0$}}{{$\frac{1}{3} - \frac{\eps}{3}$\ }{$0$}{$\frac{1}{3}$}}
	\payoffpairs{2}{{$0$}{$c\cdot\kappa$}{$0$}}{{0}{$\frac{1}{3} - \frac{\eps}{3}$\ }{$\frac{1}{3}$}}
\payoffpairs{3}{{$0$}{$c\cdot\kappa$}{$0$}}{{0}{$\frac{1}{3} - \frac{\eps}{3}$\ }{$\frac{1}{3}$}}
\payoffpairs{4}{{$\frac{1}{3}$}{$\frac{1}{3}$}{$\frac{1}{3}$}}
{{$0$}{$0$}{$\frac{1}{3}$}}
}
}
\end{center}
\label{fig:simple}
\end{subfigure}
\caption{Left: The game between clause player $c_j$ and variable player $v_i$ for Problem~1.
Right: The game between $c_j$ and $v_i$ for
Problems 2--9. }
\label{fig:combined}

\end{figure}

\paragraph{\bf Correctness.}

Observe that the constructed game is not normalised. We prove
our result for every \emph{possible} \eps, and thus in the normalised game the result
will hold for every $\eps \in [0,1]$.
Our goal is to show that, for every possible \eps, there is an $\eps$-NE with
social welfare $m$ if and only if \ph is a ``Yes'' instance.
We begin by showing that if there is a solution for
$\phi$, then there is an \emph{exact} NE for $G$ with social welfare~$m$, and
therefore there is also an $\eps$-NE for every  possible \eps with social welfare
$m$.
We start with a simple observation about the maximum and minimum payoffs that
players can obtain in $G$.
\begin{lemma}
\label{lem:maxpayoffG}
	In $G$, the total payoff for every variable player is at most $0$, and the 
	total payoff for every clause player $c_j$ is at most $1$. Moreover, if~$c_j$
	gets payoff~$1$, then $c_j$ and the variable players connected to $c_j$ play pure 
	strategies.
\end{lemma}
\begin{lemma}
\label{lem:d-poly-compl}
If $\phi$ is a ``Yes'' instance, there is an 
NE for $G$ with social welfare~$m$.
\end{lemma}
%


We now prove that if there is a strategy profile of  
$G$ with social welfare~$m$ then~\ph is a ``Yes'' instance.
Clearly, if this statement holds for any strategy profile,
it also holds for all $\eps$-NE for any $\eps$.  
\begin{lemma}
\label{lem:d-poly-sound}
If there is a strategy profile for $G$ with social welfare $m$, then \ph is a ``Yes'' instance.
\end{lemma}

Together, Lemma~\ref{lem:d-poly-compl} and Lemma~\ref{lem:d-poly-sound} 
show that for all possible values of $\eps$, it is 
\NP-complete to decide whether there exists an $\eps$-NE for $G$ with social 
welfare $m$. When we normalise the payoffs in $[0,1]$, this holds for all $\eps \in [0,1]$.

\begin{theorem}
\label{thm:sw-hard}
Problem~\ref{probc:largetp} is \NP-complete for all $\eps \in [0,1]$, even for
degree-3 bipartite planar polymatrix games in which each player has 
at most~3 pure strategies.
\end{theorem}



\paragraph{\bf Hardness of Problems~{\ref{probc:smalltp}--\ref{probc:maxsupport}}.}
To show the hardness Problems~{\ref{probc:smalltp}--\ref{probc:maxsupport}}, 
we modify the game constructed in the previous section.
We use $G'$ to denote the new polymatrix game. The interaction graph for $G'$
is exactly the same as for the game $G$. 
The bimatrix games are extended by an extra
pure strategy for each player, the strategy Out, and the payoffs are adapted. 
If variable $x_i$ is in clause $y_j$, then 
the bimatrix game between clause player $c_j$ and $v_i$ is shown on the right in
Figure~\ref{fig:combined}.
To fix the constants, given $\eps \in (0,1)$, we choose $c$ to be in the range
$(\max(1-\frac{3\eps}{2},0),1)$, and we set $\kappa = \frac{1-\eps}{1+2c}$.
Observe that $0< c < 1$, and that $\kappa+ 2c\cdot\kappa = 1-\eps$. Furthermore,
since $c > 1-\frac{3\eps}{2}$ we have $0 < \kappa < \frac{1}{3}$. 



 
In the next lemma we prove that if \ph is a ``Yes'' instance, then the strategy
profile in which all players play according to to the assignment is an
$\eps$-WSNE. No player uses the strategy Out in this strategy profile.
%
\begin{lemma}
\label{lem:md-poly-compl}
If \ph is a ``Yes'' instance, then $G'$ possesses an
$\eps$-WSNE such that no player uses strategy Out.
\end{lemma}
On the other hand, we can prove that if
\ph is a ``No'' instance, then in every 
$\eps$-WSNE of $G'$ all players play the pure strategy Out.
\begin{lemma}
\label{lem:md-poly-sound}
If \ph is a ``No'' instance, then $G'$ 
possesses a unique $\eps$-WSNE where every player plays Out.
\end{lemma}
The combination of these two properties allows us to show that 
Problems~\ref{probc:smalltp}--\ref{probc:difcomp} are \NP-complete. For example,
for Problem~\ref{probc:restrictedweak}, we can ask whether there is an
$\eps$-WSNE of the game in which player one does not player Out.
\begin{theorem}
\label{thm:norm-NP}
Problems~\ref{probc:smalltp}--\ref{probc:difcomp}
are \NP-complete for all $\eps \in (0,1)$, even on degree-3 planar bipartite
polymatrix games where each player has at most 4 pure strategies.
\end{theorem}

\paragraph{\bf Duplicating strategies.}

To show hardness for Problems~\ref{probc:maxprob}--\ref{probc:maxsupport}, we
slightly modify the game $G'$ by duplicating every pure strategy except Out for all of the players. Since each player $c_j \in C$ has the pure
strategies $i, k, l$ and Out, we give player $c_j$ pure strategies $i', k'$ and
$l'$, which each have identical payoffs as the original strategies. Similarly
for each player $v_i \in V$ we add the pure strategies $\text{True}'$ and
$\text{False}'$. Let us denote the game with the duplicated strategies by
$\widetilde{G}$. 
Then, if \ph is a ``Yes'' instance,
we can construct an $\eps$-WSNE in which no player plays
Out, where each player places at most $0.5$ probability on each pure
strategy, and where each player uses a support of size 2. These properties are
sufficient to show that Problems \ref{probc:maxprob}--\ref{probc:maxsupport} are
\NP-complete.

\begin{theorem}
\label{thm:support-NP}
Problems~\ref{probc:maxprob}--\ref{probc:maxsupport}
are \NP-complete for 
all ${\eps \in (0,1)}$, even on degree-3 planar bipartite
polymatrix games where each player has at most 7 pure strategies.
\end{theorem}

%

\section{Constrained equilibria in bounded treewidth games}

In this section we show that some constrained equilibrium problems can be solved
in quasi-polynomial time if the input game has bounded treewidth and at most
logarithmically many actions per player. We first present a 
dynamic programming algorithm for finding approximate Nash equilibria in these
games,
and then show that it can be modified to find 
constrained equilibria.


\paragraph{\bf Tree Decompositions.}
A tree decomposition of a graph $G=(V,E)$ is a pair $(\xcal,T)$, where $T=(I,F)$
is a tree and $\xcal = \{X_i| i \in I\}$ is a family of subsets of $V$ such that
(1) $\bigcup_{i \in I} X_i = V$, 
(2) for every edge $(u,v) \in E$ there exists an $i \in I$ such that 
$\{u, v \} \in X_i$, and
(3)
for all $i,j,k \in I$ if $j$ is on the path from $i$ to $k$ in $T$, then
$X_i \cap X_k \subseteq X_j$.
The width of a tree decomposition $(\xcal,T)$ is $\max_i |X_i| - 1$.
The treewidth of a graph is the minimum width over all possible tree 
decompositions of the graph.
In general, computing the treewidth of a graph is \NP-hard, but there are fixed
parameter tractable algorithms for the problem. In particular Bodlaender~\cite{bodlaender96} has
given an algorithm that runs in $O(f(w) \cdot n)$ time, where $w$ is the
treewidth of the graph, and $n$ is the number of nodes. 


\subsection{An algorithm to find approximate Nash equilibria}
\label{sec:algbase}

Let $G$ be a polymatrix game and let $(\xcal,T)$ be a tree decomposition of $G$'s interaction graph.
We assume that an arbitrary node of $T$ has been chosen as the
root. Then, given some node $v$ in $T$, we define $G(X_v)$ to be the subgame
that is obtained when we only consider the players in the subtree of $v$. More
formally, this means that we only include players $i$ that are contained in some
set $X_u$ where $u$ is in the subtree of $v$ in the tree decomposition.
Furthermore, we will use $\tilg(X_v)$ to denote the players in $G(X_v)\setminus
X_v$. For every player $i \in X_v$, we will use $N_i(X_v)$ to denote the
neighbours of $i$ in $\tilg(X_v)$.



\paragraph{\bf $k$-uniform strategies.}

A strategy $s$ is said to be $k$-uniform if there exists a multi-set $S$ of $k$
pure strategies such that $s$ plays a uniformly over the pure strategies in $S$.
These strategies naturally arise when we sample, with replacement, $k$ pure
strategies from a distribution, and play the sampled strategies uniformly. 
%
The following is a theorem of \cite{BBP14}.
\begin{theorem}
\label{thm:small}
Every $n$-player $m$-action game has a $k$-uniform $\eps$-NE whenever 
$k \ge 8 \cdot \frac{\ln m + \ln n - \ln \eps + \ln 8}{\eps^2}$.
\end{theorem}

\paragraph{\bf Candidates and witnesses.}

For each node $v$ in the tree decomposition, we compute a set
of \emph{witnesses}, where each witness corresponds to an $\eps$-NE in $G(X_v)$.
Our witnesses have two components: $\sbf$ provides a $k$-uniform strategy
profile for the players in $X_v$, while $\payv$ contains information about the
payoff that the players in $X_v$ obtain from the players in $\tilg(X_v)$. By
summarising the information about the players in $\tilg(X_v)$, we are able to
keep the number of witnesses small.

There is one extra complication, however, which is that the number of possible
payoff vectors that can be stored in $\payv$ depends on the number of different
strategies for the players in $\tilg(X_v)$, which is exponential, and will cause
our dynamic programming table to be too large. To resolve this, we \emph{round}
the entries of $\payv$ to a suitably small set of rounded payoffs.

Formally, we first define
$P = \{ x \in [0,1] \; : \; x = \frac{\eps}{2n} \cdot k \text{ for some } k \in
\mathbb{N}\},$
to be the set of rounded payoffs. 
Then, given a node $v$ in the tree decomposition, we say that a tuple $(\sbf,
\payv)$ is a \emph{$k$-candidate} if:
\begin{itemize}
\item $\sbf$ is a set of strategies of size $|X_v|$, which gives one
strategy for each player in $X_v$.
\item Every strategy in $\sbf$ is $k$-uniform.
\item $\payv$ is a set of payoff vectors of size $|X_v|$. Each element  
$\payv_{i} \in \payv$ is of the form $P^m$, and assigns a rounded payoff to each
pure strategy of player $i$.
\end{itemize}
The set of candidates gives the set of possible entries that can appear in our
dynamic programming table. Every witness is a candidate, but not every
candidate is a witness.
The total number
of $k$-candidates for each tree decomposition node $v$ can be derived as follows. Each player has $m^k$
possible $k$-uniform strategies, and so there are $m^{kw}$ possibilities for
$\sbf$. We have that $|P| = \frac{2n}{\eps}$, and that $\payv$ contains $m \cdot
w$ elements of $P$, so the total number of possibilities for $\payv$ is
$(2\cdot \frac{n}{\eps})^{mw}$. Hence, the total number of candidates
for $v$ is $m^{kw} \cdot (2\cdot \frac{n}{\eps})^{mw}$.

Next, we define what it means for a candidate to be a witness. We say that a
$k$-candidate is an \emph{$\eps,k,r$-witness} if there exists a profile
$\sbf'$ for $G(X_v)$ where
\begin{itemize}
\item $\sbf'$ agrees with $\sbf$ for the players in $X_v$.

\item Every player in $\tilg(X_v)$ is \emph{$\eps$-happy}, which means that no
player in $\tilg(X_v)$ can increase their payoff by more than $\eps$ by
unilaterally deviating from $\sbf'$. Note that this does not apply to the
players in $X_v$.

\item Each payoff vector $p \in \payv$ is within $r$ of the payoff that player $i$
obtains from the players in $\tilg(X_v)$. More accurately, for every pure
strategy $l$ of player $i$ we have that:
$\lVert \; p_l - \sum_{j \in \tilg(X_v)} (A_{ij} \cdot \sbf'_j)_l \; \rVert \le r.$
Note that $\payv$ does not capture the payoff obtained from players in
$X_v$, only those in the subtree of $v$.
\end{itemize}

\paragraph{\bf The algorithm.}
Our algorithm computes a set of witnesses for each tree
decomposition node by dynamic programming. At every leaf, the algorithm checks
every possible candidate to check whether it is a witness. At internal nodes in
the tree decomposition, if a vertex is \emph{forgotten}, that is, if it appears
in a child of a node, but not in the node itself, then we use the set of
witnesses computed for the child to check whether the forgotten node is
$\eps$-happy. If this is the case, then we create a corresponding witness for
the parent node. The complication here is that, since we use rounded payoff
vectors, this check may declare that a player is $\eps$-happy erroneously due to
rounding errors. So, during the analysis we must be careful to track the total
amount of rounding error that can be introduced.

Once a set of witnesses has been computed for every tree decomposition node, a
second phase is then used to find an $\eps$-NE of the game. This phase picks an
arbitrary witness in the root node, and then unrolls it by walking down the
tree decomposition and finding the witnesses that were used to generate it.
These witnesses collectively assign a $k$-uniform strategy profile to each
player, and this strategy profile will be the $\eps$-NE that we are looking for.

Due to space constraints, we do not give a full description of the algorithm
here, but we give a complete specification in Appendix~\ref{app:alg}. The
following lemma summarises the key properties of the algorithm.

\begin{lemma}
\label{lem:alg}
There is a dynamic programming algorithm that runs in time 
$O(n \cdot m^{2kw} \cdot (\frac{n}{\eps})^{2mw})$ that, for each tree
decomposition node $v$, computes a set of candidates $C(v)$ such that:
\begin{itemize}
\item 
Every candidate $(\sbf, \payv) \in C(v)$ is an $\eps_v,k,r_v$-witness for $v$
for some $\eps_v \le 1.5\eps$ and $r_v \le \frac{\eps}{4}$ 
\item 
If $\sbf$ is a $k$-uniform $\eps/4$-NE then $C(v)$ will contain a witness
$(\sbf', \payv)$ such that $\sbf'$ agrees with $\sbf$ for all players in $X_v$.
\end{itemize}
\end{lemma}
The running time bound arises from the total number of possible candidates for
each tree decomposition node. 
The first property ensures
that the algorithm always produces a $1.5\eps$-NE of the game, provided that
the root node contains a witness. The second property ensures that the root node
will contain a witness provided that game has a $k$-uniform $\eps/4$-NE.
Theorem~\ref{thm:small} tells us how large $k$ needs to be for this to be the
case. These facts yields the following theorem.

\begin{theorem}
\label{thm:twmain}
Let $\eps > 0$,
$G$ be a polymatrix game with treewidth $w$,
  and $k = 128 \cdot \frac{\ln m + \ln n - \ln \eps + \ln
8}{\eps^2}$. There is an algorithm that 
finds a $1.5\eps$-NE of $G$ in 
in $O(n \cdot m^{2kw} + 
(\frac{n}{\eps})^{2mw})$ time. 
\end{theorem}

\noindent Note that if $m \le \ln n$ (and in particular if $m$ is
constant), this is a QPTAS.

\begin{corollary}
Let $\eps > 0$, and
$G$ be a polymatrix game with treewidth $w$, and $m \le \ln n$. 
There is an algorithm that finds a 
$1.5\eps$-NE of $G$ in 
$(\frac{n}{\eps})^{O(\frac{w \cdot \ln n}{\eps^2})}$
time. 
\end{corollary}


\subsection{Constrained approximate Nash equilibria}

\paragraph{\bf One variable decomposable constraints.}

We now adapt the algorithm to find a certain class of constrained approximate
Nash equilibria. As a motivating example, consider Problem~\ref{probc:largetp},
which asks us to find an approximate NE with high social welfare. Formally, this
constraint assigns a single rational number (the social welfare) to each
strategy profile, and asks us to maximize this number. This constraint also
satisfies a \emph{decomposability} property: if a game $G$ consists of two
subgames $G_1$ and $G_2$, and if there are no edges between these two subgames,
then we can maximize social welfare in $G$ by maximizing social welfare in $G_1$
and $G_2$ independently. 
%
%
We formalise this by defining a constraint to be \emph{one variable decomposable
(OVD)} if the following conditions hold.
\begin{itemize}
\item There is a polynomial-time computable function $g$ such that maps every
strategy profile in $G$ to a rational number.
\item Let $\sbf$ be a strategy for game $G$, and suppose that we want to add
vertex~$v$ to $G$. Let $s$ be a strategy choice for $v$, and $\sbf'$ be an
extension of $\sbf$ that assigns $s$ to $v$. There is a polynomial-time
computable function $\addf$ such that $g(\sbf') = \addf(G, v, s, g(\sbf))$.
\item Let $G_1$ and $G_2$ be two subgames that partition $G$, and suppose that
there are no edges between $G_1$ and $G_2$. Let $\sbf_1$ be a strategy profile in $G_1$ and
$\sbf_2$ be a strategy profile in $G_2$. If $\sbf$ is the strategy profile for $G$
that corresponds to merging $\sbf_1$ and $\sbf_2$, then there is a
polynomial-time computable function $\mergef$ such that $g(\sbf) = \mergef(G_1, G_2, g(\sbf_1), g(\sbf_2)).$
\end{itemize}
Intuitively, the second condition allows us to add a new vertex to a subgame,
and the third condition allows us to merge two disconnected subgames. 
Moreover, observe that the functions $\addf$ and $\mergef$
depend only on the value that $g$ assigns to strategies, and not the strategies
themselves. This allows our algorithm to only store the value assigned by $g$,
and forget the strategies themselves.

\paragraph{\bf Examples of OVD constraints.}

Many of the problems in Table~\ref{tab:problems} are OVD
constraints. Problems~\ref{probc:largetp} and \ref{probc:smalltp} refer to the
total payoff of the strategy profile, and so $g$ is defined to be the total
payoff of all players, while the functions $\addf$ and $\mergef$ simply add the total payoff
of the two strategy profiles. 
Problems~\ref{probc:smallp} and~\ref{probc:maxprob} both deal with minimizing a
quantity associated with a strategy profile, so for these problems the functions
$\addf$ and $\mergef$ use the $\min$ function to minimize the relevant
quantities. Likewise, Problems~\ref{probc:totalmaxsupport},~\ref{probc:minmaxsupport}, 
and~\ref{probc:maxsupport} seek to maximize a quantity, and so the functions 
$\addf$ and $\mergef$ use the $\max$ function. In all cases, proving 
the required
properties for the functions is straightforward. 










\paragraph{\bf Finding OVD $k$-uniform constrained equilibria.}

We now show that, for every OVD constraint, the algorithm presented in Section~\ref{sec:algbase} can be
modified to find a $k$-uniform $1.5\eps$-NE that also has a high value with
respect to the constraint. More formally, we show that the value assigned by $g$
to the $1.5\eps$-NE is greater than the value assigned to $g$ to all $k$-uniform
$\eps/4$-NE in the game.

Given an OVD constraint defined by $g$, $\addf$,
and $\mergef$, we add an extra element to each candidate to track the
variable from the constraint: each candidate has the form $(\sbf, \payv, x)$,
where $\sbf$ and $\payv$ are as before, and $x$ is a rational number. The definition of an
$\eps,k,r,g$-witness is extended by adding the condition:
\begin{itemize}
\item 
Recall that $\sbf'$ is a strategy profile for $G(X_v)$ whose existence is
asserted by the witness.
Let $\sbf''$ be the restriction of $\sbf'$ to $\tilg(X_v)$. We have $x
= g(\sbf'')$.
\end{itemize} 
We then modify the algorithm to account for this new element in the witness. 
At each stage we track the value correct value for $x$. At the leafs, we use $g$
to compute the correct value. At internal nodes, we use $\addf$ and $\mergef$ to
compute the correct value using the values stored in the witnesses of the
children. 

If at any point two witnesses are created that agree on the $\sbf$ and $\payv$
components, but disagree on the $x$ component, then we only keep the witness whose
$x$ value is higher. This ensures that we only keep
witnesses corresponding to strategy profiles that maximize the constraint. When
we reach the root, we choose the strategy profile with maximal value for $x$ to
be unrolled in phase 2. The fact that we only keep one witness for each pair
$\sbf$ and $\payv$ means that the running time of the algorithm is unchanged.
Again, we defer the technical description of the algorithm to
Appendix~\ref{app:consk}, but the following theorem summarises the results.

\begin{theorem}
\label{thm:consk}
For every $\eps > 0$ let $k = 128 \cdot \frac{\ln m + \ln n - \ln \eps + \ln
8}{\eps^2}$. If $G$ is a polymatrix game with treewidth $w$, then 
there is an algorithm that runs in $O(n \cdot m^{2kw} + 
(\frac{n}{\eps})^{2mw})$ time and finds a $k$-uniform $1.5\eps$-NE $\sbf$ such
that $g(\sbf) \ge g(\sbf')$ for every strategy profile $\sbf'$ that is an
$\eps/4$-NE.
\end{theorem}

\paragraph{\bf Results for non-$k$-uniform strategies.}

The guarantee given by Theorem~\ref{thm:consk} is given relative to the best
value achievable by a $k$-uniform $\eps/4$-NE. It is also interesting to ask
whether we can drop the $k$-uniform constraint. In the following theorem, we
show that if $g$ is defined to be a linear function of the payoffs in the game,
then a guarantee can be given relative to \emph{every} $\eps/8$-NE of the
game. Note that this covers Problems~\ref{probc:largetp},~\ref{probc:smalltp},
and~\ref{probc:smallp}.

\begin{theorem}
\label{thm:nouniform}
Suppose that, for a given a OVD constraint, the function $g$ is a linear
function of the payoffs. 
Let $\sbf$ be the $1.5\eps$-NE found by our algorithm when
For every $\eps/8$-NE $\sbf'$ we have that $g(\sbf) \ge g(\sbf') - O(\eps)$.
\end{theorem}

\section{Further work}

There are two clear directions for further work: Can we extend \NP-hardness of 
Problems 2--9 to \eps-NE? We believe hardness will not hold for all $\eps \in (0,1)$ as 
for our results, but will hold for all $\eps$ less than a constant. Secondly,
can we extend the family of constraints for which we can efficiently find constrained
approximate equilibria that compare well with all other (non-$k$-uniform) approximate 
equilibria, beyond the constraints we can already deal with?

{\small 
\bibliographystyle{abbrv}
\bibliography{references}

\begin{thebibliography}{10}

\bibitem{ABC13}
P.~Austrin, M.~Braverman, and E.~Chlamtac.
\newblock Inapproximability of {NP}-complete variants of {N}ash equilibrium.
\newblock {\em Theory of Computing}, 9:117--142, 2013.

\bibitem{BBP14}
Y.~Babichenko, S.~Barman, and R.~Peretz.
\newblock Simple approximate equilibria in large games.
\newblock In {\em EC}, pages 753--770, 2014.

\bibitem{BL15}
S.~Barman and K.~Ligett.
\newblock Finding any nontrivial coarse correlated equilibrium is hard.
\newblock In {\em Proc.\ of EC}, pages 815--816, 2015.

\bibitem{BLP15}
S.~Barman, K.~Ligett, and G.~Piliouras.
\newblock Approximating {N}ash equilibria in tree polymatrix games.
\newblock In {\em Proc.\ of SAGT}, pages 285--296, 2015.

\bibitem{BM16}
V.~Bil{\`{o}} and M.~Mavronicolas.
\newblock A catalog of {$\exists\reals$}-complete decision problems about
  {N}ash equilibria in multi-player games.
\newblock In {\em Proc.\ of {STACS}}, pages 17:1--17:13, 2016.

\bibitem{BM17}
V.~Bil{\`{o}} and M.~Mavronicolas.
\newblock {$\exists\reals$}-complete decision problems about symmetric {N}ash
  equilibria in symmetric multi-player games.
\newblock In {\em Proc.\ of {STACS}}, pages 13:1--13:14, 2017.

\bibitem{bodlaender96}
H.~L. Bodlaender.
\newblock A linear-time algorithm for finding tree-decompositions of small
  treewidth.
\newblock {\em SIAM Journal on Computing}, 25(6):1305--1317, 1996.

\bibitem{BBM10}
H.~Bosse, J.~Byrka, and E.~Markakis.
\newblock New algorithms for approximate {N}ash equilibria in bimatrix games.
\newblock {\em Theoretical Computer Science}, 411(1):164--173, 2010.

\bibitem{BKW15}
M.~Braverman, Y.~Kun{-}Ko, and O.~Weinstein.
\newblock Approximating the best {N}ash equilibrium in
  \emph{n\({}^{\mbox{o}}\)}\({}^{\mbox{(log \emph{n})}}\)-time breaks the
  exponential time hypothesis.
\newblock In {\em Proc.\ of SODA}, pages 970--982, 2015.

\bibitem{CDT09}
X.~Chen, X.~Deng, and S.-H. Teng.
\newblock Settling the complexity of computing two-player {N}ash equilibria.
\newblock {\em Journal of the ACM}, 56(3):14:1--14:57, 2009.

\bibitem{CS08}
V.~Conitzer and T.~Sandholm.
\newblock New complexity results about {N}ash equilibria.
\newblock {\em Games and Economic Behavior}, 63(2):621 -- 641, 2008.

\bibitem{CDFFJS}
A.~Czumaj, A.~Deligkas, M.~Fasoulakis, J.~Fearnley, M.~Jurdzi{\'n}ski, and
  R.~Savani.
\newblock Distributed methods for computing approximate equilibria.
\newblock In {\em Proc.\ of WINE}, pages 15--28, 2016.

\bibitem{CFJ15}
A.~Czumaj, M.~Fasoulakis, and M.~Jurdzi{\'n}ski.
\newblock Approximate {N}ash equilibria with near optimal social welfare.
\newblock In {\em Proc.\ of IJCAI}, pages 504--510, 2015.

\bibitem{CFJ16}
A.~Czumaj, M.~Fasoulakis, and M.~Jurdzi{\'n}ski.
\newblock Approximate plutocratic and egalitarian {N}ash equilibria.
\newblock In {\em Proc.\ of AAMAS}, pages 1409--1410, 2016.

\bibitem{DGP09}
C.~Daskalakis, P.~W. Goldberg, and C.~H. Papadimitriou.
\newblock The complexity of computing a {N}ash equilibrium.
\newblock {\em SIAM Journal on Computing}, 39(1):195--259, 2009.

\bibitem{DMP07}
C.~Daskalakis, A.~Mehta, and C.~H. Papadimitriou.
\newblock Progress in approximate {N}ash equilibria.
\newblock In {\em Proc.\ of EC}, pages 355--358, 2007.

\bibitem{DMP09}
C.~Daskalakis, A.~Mehta, and C.~H. Papadimitriou.
\newblock A note on approximate {N}ash equilibria.
\newblock {\em Theoretical Computer Science}, 410(17):1581--1588, 2009.

\bibitem{DFIS}
A.~Deligkas, J.~Fearnley, T.~P. Igwe, and R.~Savani.
\newblock An empirical study on computing equilibria in polymatrix games.
\newblock In {\em Proc.\ of AAMAS}, pages 186--195, 2016.

\bibitem{DFS16}
A.~Deligkas, J.~Fearnley, and R.~Savani.
\newblock Inapproximability results for approximate {N}ash equilibria.
\newblock In {\em Proc.\ of WINE}, pages 29--43, 2016.

\bibitem{polymatrix}
A.~Deligkas, J.~Fearnley, R.~Savani, and P.~G. Spirakis.
\newblock Computing approximate {N}ash equilibria in polymatrix games.
\newblock {\em Algorithmica}, 77(2):487--514, 2017.

\bibitem{FGSS12}
J.~Fearnley, P.~W. Goldberg, R.~Savani, and T.~B. S{\o}rensen.
\newblock Approximate well-supported {N}ash equilibria below two-thirds.
\newblock In {\em SAGT}, pages 108--119, 2012.

\bibitem{GMVY}
J.~Garg, R.~Mehta, V.~V. Vazirani, and S.~Yazdanbod.
\newblock Etr-completeness for decision versions of multi-player (symmetric)
  {N}ash equilibria.
\newblock In {\em Proc.\ of ICALP}, pages 554--566, 2015.

\bibitem{GZ89}
I.~Gilboa and E.~Zemel.
\newblock {N}ash and correlated equilibria: Some complexity considerations.
\newblock {\em Games and Economic Behavior}, 1(1):80 -- 93, 1989.

\bibitem{HK11}
E.~Hazan and R.~Krauthgamer.
\newblock How hard is it to approximate the best {N}ash equilibrium?
\newblock {\em {SIAM} J. Comput.}, 40(1):79--91, 2011.

\bibitem{Kloks}
T.~Kloks.
\newblock {\em Treewidth, Computations and Approximations}, volume 842 of {\em
  Lecture Notes in Computer Science}.
\newblock Springer, 1994.

\bibitem{KS10}
S.~C. Kontogiannis and P.~G. Spirakis.
\newblock Well supported approximate equilibria in bimatrix games.
\newblock {\em Algorithmica}, 57(4):653--667, 2010.

\bibitem{LMM03}
R.~J. Lipton, E.~Markakis, and A.~Mehta.
\newblock Playing large games using simple strategies.
\newblock In {\em Proc.\ of EC}, pages 36--41, 2003.

\bibitem{MR01}
C.~Moore and J.~M. Robson.
\newblock Hard tiling problems with simple tiles.
\newblock {\em Discrete {\&} Computational Geometry}, 26(4):573--590, 2001.

\bibitem{OI16}
L.~E. Ortiz and M.~T. Irfan.
\newblock {FPTAS} for mixed-strategy {N}ash equilibria in tree graphical games
  and their generalizations.
\newblock {\em CoRR}, abs/1602.05237, 2016.

\bibitem{R15}
A.~Rubinstein.
\newblock Inapproximability of {N}ash equilibrium.
\newblock In {\em {STOC}}, pages 409--418, 2015.

\bibitem{TS08}
H.~Tsaknakis and P.~G. Spirakis.
\newblock An optimization approach for approximate {N}ash equilibria.
\newblock {\em Internet Mathematics}, 5(4):365--382, 2008.

\end{thebibliography}
}

\appendix
\newpage

\section{Proof of Lemma~\ref{lem:maxpayoffG}}

\begin{proof}
That every variable player gets at most $0$ in total follows from the fact that
the largest payoff entry for the variable player in the bimatrix games is 0. 

Now consider a clause player $c_j$ that interacts with variable players $i, k$ and~$l$. 
Denote by $p_i, p_k$ and $p_l$ the probabilities that player $i, k$ and~$l$ play
the pure strategy True respectively. Pure strategy $i$ yields payoff 
$p_i-p_l-p_k$ for player $c_j$,~$l$ yields payoff $-p_i+p_l-p_k$,
	and $k$ yields payoff $-p_i-p_l+p_k$. Since $p_i, p_k, p_l \in [0,1]$,
the maximum value for each expression is $1$. 
Moreover, the maximum value of $1$ is achieved only when exactly one of the
variable players plays his pure strategy True with probability 1 and the
other two players play their pure strategy False, and thus at most one of these
expressions if $1$, so the clause player must player a pure strategy to get $1$ in total.
\qed
\end{proof}

\section{Proof of Lemma~\ref{lem:d-poly-compl}}

\begin{proof}
A 1-in-3 satisfying assignment for $\phi$ directly corresponds to a pure
	strategy profile for~$G$: variable player $v_i \in V$ plays the truth value of $x_i$; 
	clause player~${c_j \in C}$ plays the pure strategy that corresponds to the
	single true member variable of~$y_j$. 
We call such a strategy profile $s$.
We prove that $s$ is an NE for~$G$ and that its social welfare is $m$.

First, we show that $s$ is an NE.
Observe that each player in $C$ gets payoff 1 under $s$.
Since, by Lemma~\ref{lem:maxpayoffG}, this is best payoff that a clause player
can obtain,
	no player from $C$ has a reason to deviate.
Now consider a variable player $v_i \in V$. 
If $v_i$ plays the pure strategy True, then he gets a total payoff of 0 irrespective of what the clause
players do. 
On the other hand, if player $v_i$ plays the pure strategy False in $s$, then
	he gets payoff $0$ from every bimatrix game he
	participates in because clause players only play strategy $i$ if $v_i$ plays
	True. 
Thus, in either case, player $v_i$ cannot gain by
deviating, since, by Lemma~\ref{lem:maxpayoffG}, $0$ is the largest payoff that~$v_i$ can obtain.

Finally, observe that the social welfare of the strategy profile is $m$, because
as we have argued, every clause player gets $1$ and every variable
player gets $0$.\qed
\end{proof}

\section{Proof of Lemma~\ref{lem:d-poly-sound}}

\begin{proof}
Suppose that there is a strategy profile $s$ with social welfare $m$. 
By Lemma~\ref{lem:maxpayoffG}, in order to achieve social welfare $m$, every 
clause player must get payoff $1$,
and this is only possible under pure strategy profiles. 
So, $s$ is a pure strategy profile, and naturally defines an assignment to the
variables in \ph according to the pure strategies played by the variable players.
We argue that this is a satisfying 1-in-3 assignment of \ph.
The reason is that, when clause player $c_j$ plays strategy $i$, variable player
$v_i$ must pick True, and the other two players must pick False, because
otherwise $c_j$ would not obtain payoff~$1$. Since this holds for all clauses,
and every clause gets payoff 1, every clause must have exactly one true literal.
\qed
\end{proof}

\section{Proof of Lemma~\ref{lem:md-poly-compl}}

\begin{proof}
In exactly the same way as in the proof of Lemma~\ref{lem:d-poly-compl}, 
we interpret an \mbox{1-in-3} satisfying assignment for $\phi$ as a pure
strategy profile $s$ for~$G$.
We must argue that such an~$s$ is an \eps-WSNE.
Under $s$, every clause player gets payoff~${\kappa+2c\cdot\kappa}$; he gets
$\kappa$ from the game he plays with the player that plays the pure strategy
True and $c \cdot \kappa$ from each of the two games he plays with the players
that play the pure strategy False.
The expected payoff from 
Out is 1 (because every clause player has degree exactly~3). 
Thus, every clause player plays a pure strategy that is a
$(1 - \kappa - 2c \cdot \kappa)$-best response, which is an $\eps$-best response
by the choice of $\kappa$ and $c$.
Under $s$, every variable player $v_i$ gets payoff $1-\eps$ since $v_i$
plays three bimatrix games with clause players that either all play $i$ if $v_i$
	plays True or none of them play $i$ if $v_i$ plays False.
The maximum payoff that $v_i$ could get is 1 from pure
strategy Out.
Thus, every variable player plays a pure strategy that is a $\eps$-best
response. 
Hence, the constructed strategy profile is a $\eps$-WSNE. 
\qed
\end{proof}

\section{Proof of Lemma~\ref{lem:md-poly-sound}}

\begin{proof}
Assume that \ph is a ``No'' instance. 
We show that in this case in every $\eps$-WSNE
there is at least one player who plays Out as a pure strategy. 
Towards a contradiction, assume that no clause player plays Out, and 
consider a clause player $c_j$, who is connected to $i$, $k$, and $l$.
The maximum payoff that $c_j$ can get is $\kappa+2c\cdot\kappa = 1-\eps$, if and 
only if exactly one of $v_i$, $v_k$, and $v_l$ plays True and the other
two play False.
If every clause player could achieve payoff $1-\eps$, we would 
then have an 1-in-3 satisfying assignment for \ph, which would be a contradiction.
Thus at least one clause player, say $c_j$, gets payoff less than $1-\eps$.
However, if $c_j$ has payoff strictly less than $1-\eps$ and does not play
Out, then we do not have an $\eps$-WSNE, since Out always gives $c_j$ 
payoff 1.   

Having established that in any $\eps$-WSNE there exists at least one 
player $c_j \in C$ that plays Out as a pure strategy, we now prove that 
all other players must also play Out as a pure strategy. There are two cases to
consider.
\begin{itemize}
\item Let $v_i$ be a variable player, $c_j$ be a clause player who has an
edge to $v_i$, and suppose that $c_j$ plays Out as a pure strategy. Observe
that, if $v_i$ plays either True or False, his payoff can be at most $\frac{2}{3}-\frac{2\eps}{3}$,
whereas he can always obtain payoff 1 from
playing Out, since at least one of his neighbours plays Out by assumption. 
So $v_i$ must play Out
as a pure strategy in order to be in a $\eps$-WSNE.
\item Let $v_i$ be a variable player, $c_j$ be a clause player who has an edge
to $v_i$, and suppose that $v_i$ plays Out as a pure strategy. Observe that if
this is the case, then $c_j$ cannot obtain payoff $\kappa+2c\cdot\kappa= 1-\eps$ from the strategies
$i$, $j$, and $k$. On the other hand, $c_j$ obtains payoff 1 from Out
under all strategies profiles, and so $c_j$ must play Out as a pure strategy in
order to be in a $\eps$-WSNE.
\end{itemize}
By iteratively applying these two arguments we can prove that, since
there is at least one player playing Out, all other players must also play Out. Hence, the
strategy profile in which every player plays Out is the unique 
$\eps$-WSNE of the game. \qed
\end{proof}

\section{Proof of Theorem~\ref{thm:norm-NP}}

\begin{proof}
From Lemmas~\ref{lem:md-poly-compl} and~\ref{lem:md-poly-sound} 
we can get the following two facts 
about $G'$.
If \ph is a ``Yes'' instance, then $G'$ 
possesses an $\eps$-WSNE such that:
\begin{itemize}
\item [(a)] no player plays the pure strategy Out,
\item [(b)] every player gets payoff $1-\eps$.
\end{itemize}
If \ph is a ``No'' instance, then for every
$\eps$-WSNE 
of $G'$ we have that:
\begin{itemize}
\item [(i)] every player places all probability on the strategy Out,
\item [(ii)] every player gets payoff 1.
\end{itemize}
We will use these properties to show that each of the problems is \NP-complete.
\begin{description}
\item[Problem~\ref{probc:smalltp}] is \NP-complete because of (b) with (ii) when we set 
$u = (m+n)\cdot (1-\eps)$, where $n$ and $m$ are the 
number of variables and clauses in \ph respectively.
This is because, if \ph is a ``Yes'' instance, from (b) we get that there exists an 
$\eps$-WSNE where the sum of players payoffs sums up to $u$, while if 
\ph is a ``No'' instance, then in the unique $\eps$-WSNE of the game the
players' payoffs sum up to $n+m$.

\item[Problem~\ref{probc:smallp}] with $u=1-\eps$ is also \NP-complete because 
of (b) and (ii).

\item[Problem~\ref{probc:restrictedweak}]  is \NP-complete because of (a) and (i).
We focus on a variable player and set $S = \{\text{True}, \text{False}\}$. 
If \ph is a ``Yes'' instance, then from (a) there is an $\eps$-WSNE where the variable player plays strategies
only from $S$. If \ph is a ``No'' instance, then from (i) in the unique
\eps-WSNE the variable player must play $\text{Out}$.

\item[Problem~\ref{probc:difcomp}] is \NP-complete because of (a) and (i) when 
we set $d=1$. Observe that the strategy profile where every player plays his pure strategy Out
is an exact NE, so it is an \eps-WSNE irrespective of whether \ph is a ``Yes'' or ``No'' instance.
So, from (a) we know that when \ph is a ``Yes'' instance, there is
another \eps-WSNE where no player plays his pure strategy Out, 
and there are two WSNEs with TV distance $1$. On the other hand, if \ph is a ``No'' instance, 
the ``all Out'' strategy profile is the unique \eps-WNSE. 
\end{description}
\qed
\end{proof}

\section{Proof of Theorem~\ref{thm:support-NP}}

\begin{proof}
Lemma~\ref{lem:md-poly-compl} holds also for the game~$\widetilde{G}$,
since it differs from $G'$ only in the duplication of pure strategies other than Out.
Extending the claim of Lemma~\ref{lem:md-poly-compl}, we have that if \ph 
is a ``Yes'' instance, then there is a
\eps-WSNE in 
$\widetilde{G}$ in which no player places probability on Out, and in which the
amount of probability on each pair of duplicate strategies is split evenly, i.e., 
probability one half on each, and $|\supp(s_j)| = 2$ for every player $j$ of the game.
If \ph is a ``No'' instance, then for every
\eps-WSNE 
of $\widetilde{G}$ every player places all probability on the strategy Out, so 
$|\supp(s_j)| = 1$ for every player $j$ of the game.

These properties then imply that our problems are \NP-complete.
\begin{description}
\item[Problem~\ref{probc:maxprob}] is \NP-complete when we set $p=\frac{1}{2}$.

\item[Problem~\ref{probc:totalmaxsupport}] is \NP-complete when we set $k=2(n+m)$, where
$n$ and $m$ are the number of variables and clauses in \ph respectively.
If \ph is a ``Yes'' instance, then there is an \eps-WSNE where the support sizes
of the played strategies sum up to $k$. If on the other hand the instance does not have a
solution, then in the unique \eps-WSNE of the game the support sizes sum up to
$m+n < k$.

\item[Problem~\ref{probc:minmaxsupport}] is \NP-complete when we set $k=2$.

\item[Problem~\ref{probc:maxsupport}] is \NP-complete when we set $k=2$.
\end{description}
\qed
\end{proof}

\section{Proof of Lemma~\ref{lem:alg}}
\label{app:alg}

In this section, we fully describe the dynamic programming algorithm that we
referred to in Section~\ref{sec:algbase}. The algorithm will proceed in two
phases. Phase 1 will compute a set of candidates for each node in the tree
decomposition, by starting at the leafs and working upwards. Phase 2 will then
use these sets to compute an approximate Nash equilibrium of the game.

\paragraph{\bf Nice tree decompositions.}

It has been shown that we can restrict ourselves to only considering \emph{nice}
tree decompositions~\cite{Kloks}. A tree decomposition $(\xcal,T)$ is nice if
the following conditions are satisfied:
\begin{enumerate}
\item every node of $T$ has at most two children,
\item if a node $i$ has two children $j$ and $k$, then $X_i = X_j = X_k$,
\item if a node $i$ has one child $j$, then
\begin{itemize}
\item either $|X_i| = |X_j| + 1$ and $X_j \subseteq X_i$
\item or $|X_i| = |X_j| - 1$ and $X_i \subseteq X_j$.
\end{itemize}
\end{enumerate}
In a nice tree decomposition, $T$ is a rooted binary tree, and 
each node $i$ has one of the following types.
\begin{itemize}
\item \textbf{Start}. The node $i$ is a leaf in $T$.
\item \textbf{Join}. The node $i$ has exactly two children $j$ and $k$ in $T$,
and $X_i = X_j = X_k$.
\item \textbf{Introduce}. The node $i$ has exactly one child $j$ in $T$ and 
$|X_i| = |X_j| + 1$. That is, exactly one new vertex is added as we move from
$j$ to $i$.
\item \textbf{Forget}. The node $i$ has exactly one child $j$ in $T$ and 
$|X_i| = |X_j| - 1$. That is, exactly one vertex is removed as we move from $j$
to $i$.
\end{itemize}
It is a well known fact (see, e.g.,~\cite{Kloks}) that if we have a tree
decomposition of width $w$ for a given graph, then we can construct a nice tree
decomposition of width $w$ in polynomial time. Furthermore, if the original tree
decomposition had $n$ nodes, then the nice tree decomposition will have at most
$4n$ nodes. When we are working with nice tree decompositions, we can assume
that if~$r$ is the root of the tree, then $X_r = \emptyset$, since we can keep
adding forget nodes to make this the case. We will assume from now on that our
tree decomposition is nice.

\paragraph{\bf Phase 1.}
The first phase of the algorithm will compute a set of $k$-candidates $C(v)$
for each tree decomposition node $v$. We will later prove that $C(v)$ is in
fact a set of witnesses. Since we work
with nice tree decompositions, we only need to consider four types of nodes: 

\paragraph{Start.} 

If $v$ is a start node, then we produce
$C(v)$ by listing all possible $k$-uniform strategy profiles for the players in 
$X_v$, and for each combination producing a candidate $(\sbf, \payv)$ where
$\sbf$ plays the strategies, and $\payv$ assigns payoff $0$ to every pure strategy.
This can be done in $O(m^{kw})$ time.

\paragraph{Introduce.} 

If $v$ is an introduce node, then let $u$ be the child of
$v$ in $T$. We consider every $k$-uniform strategy $s_i$ of player $i$, and
every $(\sbf, \payv) \in C(u)$, and for each pair we produce $(\sbf', \payv')$
where: \begin{itemize} \item $\sbf'_i = s_i$, and $\sbf'_j = \sbf_j$ for all $j
\ne i$, \item $\payv'_i$ is the 0 vector and $\payv_j = \payv_j$ for $j \ne i$,
\end{itemize} and we add $(\sbf', \payv')$ to $C(v)$. This operation takes
$O(m^{k} \cdot |C(u)|)$ time.

\paragraph{Forget.}

If $v$ is a forget node, then let $u$ be the child of $v$ in $T$. Let $i$ be the
player who is removed by $v$. 
For each $(\sbf, \payv) \in C(u)$ we perform the following
operation. First we check that player $i$ is $\eps$-happy, which is true if
the following inequality holds:
\begin{align}
\label{eqn:forget2}
s_i \cdot \Bigg( \payv_i + \sum_{\substack{ j \in X_{v}\cap N(i)  \\ s_j \in \sbf}}
A_{ij} \cdot s_j \Bigg) \geq \max \Bigg( \payv_i + 
\sum_{\substack{ j \in X_{v}\cap N(i)  \\ s_j \in \sbf}} A_{ij} \cdot s_j \Bigg) - \eps.
\end{align} 
This inequality simply checks whether player $i$ is $\eps$-happy, using the data
about player $i$'s payoff stored in $\payv$, and the strategies chosen by the
other players in $X_v$, using the data stored in $\sbf$. 

If the test is passed, then we create $(\sbf', \payv')$ as follows:
\begin{itemize}
\item For all $j \ne i$ we set $\sbf'(j) = \sbf(j)$.
\item For all $j \ne i$ we obtain $\payv'(j)$ by computing $\payv(j) + A_{ji}
\cdot s_i$, and then rounding it to the closest element of $P$.
\end{itemize}
This operation discards the strategy and payoff vector for player $i$, and
updates the payoff vectors for the players in $X_v$ using the strategy played by
player $i$ in $\sbf$. We then add $(\sbf', \payv')$ to $C(v)$. This
operation takes $O(|C(u)|)$ time.

\paragraph{Join.}
Finally, if $v$ is a join node, then let $u_1$ and $u_2$ be the two children of
$v$ in $T$. We consider every pair of candidates $(\sbf_{u_1}, \payv_{u_1}) \in
C(u_1)$ and $(\sbf_{u_2}, \payv_{u_2}) \in C(u_2)$. For each pair, we first
check that $\sbf_{u_1}$ and $\sbf_{u_2}$ are the same, and if so we create
$(\sbf', \payv')$ as follows:
\begin{itemize}
\item $\sbf'(i) = \sbf_{u_1}(i) = \sbf_{u_2}(i)$ for all $i \in X_v$.
\item $\payv'(i) = \payv_{u_1}(i) + \payv_{u_2}(i)$ for all $i \in X_v$.
\end{itemize}
This operation copies the strategy profile, and adds the two payoff vectors.
Note that no rounding is required, since $\payv_{u_1}(i)$ and $\payv_{u_2}(i)$
are both already rounded. This operation takes $O(|C(u_1)| \cdot
|C(u_2)|)$ time.

\paragraph{\bf Running time.}

As we have argued, the total number of candidates for each tree decomposition
node can be at most $m^{kw} \cdot (\frac{n}{\eps})^{mw}$. Since there are $O(n)$
nodes in our tree decomposition, the overall running time of the algorithm is
$O(n \cdot m^{2kw} \cdot (\frac{n}{\eps})^{2mw})$.

\paragraph{\bf Correctness.}

We first prove that every candidate computed by the algorithm is a witness.
For each node $v$ in the tree decomposition, let $f(v)$ be the total number of
forget nodes over all paths from $v$ to a leaf (including $v$ itself if it is a
forget node).

\begin{lemma}
\label{lem:soundness}
For every node $v$ in the tree decomposition, let $r_v = \frac{\eps \cdot
f(v)}{4n}$ and $\eps_v = \eps + 2\cdot r_v$. Every candidate $(\sbf, \payv) \in
C(v)$ is an $\eps_v,k,r_v$-witness for $v$.
\end{lemma}
\begin{proof}
Our proof will be by induction over the nodes in the tree decomposition. For the
base case, we consider the case where $v$ is a start node. In this case, observe
that the game $\tilg(X_v)$ is empty, so a $k$-candidate $(\sbf, \payv)$ is an
$\eps_v,k,r_v$-witness if and only if $\sbf$ contains only $k$-uniform
strategies, and $\payv$ contains only $0$ vectors. Since $C(v)$ only contains
candidates that satisfy these criteria, every member of $C(v)$ is an
$\eps_v,k,r_v$-witness.

For the inductive step, there are three possibilities, depending on the type of
	$v$. In each of these cases, we assume, as inductive hypothesis, that we
	have proved that every member of $C(u)$ is an $\eps_u,k,r_u$-witness for all
	children $u$ of $v$.


\paragraph{Introduce nodes.} Let $v$ be an introduce node, and let $(\sbf,
\payv) \in C(u)$ be a witness for $u$. We will show that the candidate $(\sbf',
\payv')$ created by our algorithm is an $\eps_u,k,r_u$-witness for $v$. Observe
that, since $X_u$ is a separator in $G$, player $i$ has no edges to any player
in $\tilg(X_v)$. This means that the strategy played by player $i$ cannot affect
whether any player in $\tilg(X_v)$ is happy. Therefore, we can use the inductive
hypothesis to show that every player in $\tilg(X_v)$ is $\eps_u$-happy, and then
observe that $\eps_u = \eps_v$. Moreover, player $i$ cannot obtain any payoff
from the players in $\tilg(X_v)$, and so setting $\payv_i$ to be the all-zero vector
is correct. Thus, $(\sbf', \payv')$  is an $\eps_v,k,r_v$-witness for $v$.

\paragraph{Forget nodes.} We now consider the case where $v$ is a forget node.
Let $(\sbf, \payv)$
be a member of $C(u)$ that passes the test in
Inequality~\eqref{eqn:forget2}. Let $\sbf_f$ be the strategy profile for
$\tilg(X_u)$ in which every player is $\eps_u$-happy, whose existence is
witnessed by $(\sbf, \payv)$. We argue that if we add $s_i \in \sbf$ to $\sbf_f$,
then we obtain a strategy profile for $\tilg(X_v)$ in which every player is
$\eps_v$ happy:
\begin{itemize}
\item Every player other than $i$ continues to be $\eps_u$-happy, because by
definition $\sbf_f$ agrees with $s_i$. Since $\eps_u < \eps_v$, these players are
also $\eps_v$-happy.
\item On the other hand, we must explicitly prove that player $i$ is $\eps_v$-happy. By the inductive hypothesis, the payoffs stored in $\payv$ are within
$r_v$ of the true payoff to player $i$ under $\sbf_f$.
Thus, in the worst case, Inequality~\eqref{eqn:forget2} ensures that:
\begin{equation*}
s_i \cdot \Bigg( 
\sum_{\substack{j \in G(X_v)\cap N(i) \\ s_j \in \sbf_f}} A_{ij} \cdot s_j 
\Bigg) + r_v \geq \max \Bigg( \payv_i + 
\sum_{\substack{j \in G(X_v)\cap N(i) \\ s_j \in \sbf_f}} A_{ij} \cdot s_j \Bigg) - r_v - \eps.
\end{equation*}
This implies that player $i$ is $(\eps + 2 r_v)$-happy, as required.
\end{itemize}

Next we argue that $\payv'$, the new payoff vector constructed by our
algorithm, is correct. Consider a player $j \in X_v$. By the inductive
hypothesis, $\payv_j$ gives the payoff to $j$ from the players in $\tilg(X_u)$
with an additive error of $r_u$. We add $A_{ji} \cdot s_i$ to this vector,
obtaining the payoff to $j$ from the players in $\tilg(X_v)$ with an additive
error of $r_u$. We then round to the closest element of $P$, which
adds an additional error of at most $\eps/4n$. Since $r_v = r_u + \eps/4n$, this
operation is correct. Thus, we have shown that $(\sbf', \payv')$ is an
$\eps_v,k,r_v$-witness for $v$.

\paragraph{Join nodes.}
Finally, we consider the case where $v$ is a join node. Let $(\sbf_{u_1},
\payv_{u_1}) \in C(u_1)$ and $(\sbf_{u_2}, \payv_{u_2}) \in C(u_2)$ be a pair of
candidates. Observe that, since $X_v$ separates $\tilg(X_{u_1})$ and
$\tilg(X_{u_2})$, no player in $\tilg(X_{u_1})$ can influence a player in
$\tilg(X_{u_2})$, and vice versa. Therefore, when we merge these two witnesses
in $(\sbf', \payv')$, we do not affect whether any player in $\tilg(X_v)$ is
$\eps$-happy. Furthermore, when we add $\payv_{u_1}(i)$, which has an additive
error of $r_{u_1}$, to $\payv_{u_2}(i)$, which has an additive error of
$r_{u_2}$, then we obtain a vector that has an additive error of $r_{u_1} +
r_{u_2}$. Since $r_v$ depends on the total number of forget nodes in both
subtrees of $v$, the resulting payoff vector has an additive error of $r_v$. So,
we have shown that $(\sbf', \payv')$ is an $\eps_v,k,r_v$-witness for $v$.\qed
\end{proof}

In the other direction we must also show that the algorithm does not throw away
too many witnesses. We do this in the following lemma.

\begin{lemma}
\label{lem:completeness}
If $\sbf$ is a $k$-uniform $\eps/4$-NE then $C(v)$ will contain a witness
$(\sbf', \payv)$ such that $\sbf'$ agrees with $\sbf$ for the vertex in $X_v$.
\end{lemma}
\begin{proof}
We will show, by induction,
the stronger property that every tree node $v$ will have a witness $(\sbf, \payv)$ where:
\begin{itemize}
\item For every player $i \in X_v$ we have 
$s_i \in \sbf$ is the same as the strategy assigned to player $i$ by $\sbfne$.
\item The payoffs in $\payv$ are within $r_v = \frac{\eps \cdot f(v)}{4n}$ of
the true values given by $\tbf_v = \sum_{j \in \tilg(X_v)} A_{ij} \cdot (\sbfne)_j$.
\end{itemize}
For the base case, where $v$ is a start node, observe that by assumption $s$
only uses $k$-uniform strategies, and so $C(v)$ will contain a candidate
satisfying the required properties. For the inductive step, we again have three
possibilities, depending on the type of $v$.

If $v$ is a introduce node, then let $(\sbf, \payv)$ be the candidate for
$u$ whose existence is implied by the inductive hypothesis. Let $(\sbf',
\payv')$ be the witness that the algorithm creates using 
$(\sbf, \payv)$ and $(\sbfne)_i$. Since $X_u$ separates $i$ from every player in
$\tilg(X_v)$, we have that $\tbf_u = \tbf_v$ for all players $i \ne j$, and
$\tbf_v$ assigns the $0$ vector to player $i$. 
Since, by the inductive hypothesis, we have that $\payv$ close to $\tbf_u$ with
an additive error of at most $r_u = r_v$, the witness $(\sbf', \payv')$ created
by our algorithm does indeed satisfy the required properties.

If $v$ is a forget node, then let $(\sbf, \payv)$ be the candidate for $u$
whose existence is implied by the inductive hypothesis. We begin by arguing that
$(\sbf, \payv)$ is not discarded by the check from
Inequality~\eqref{eqn:forget2}.  Since $\sbfne$ is a
$\frac{\eps}{4}$-NE we have:
\begin{equation*}
s_i \cdot \Bigg( \tbf_i + \sum_{\substack{j \in X_{v}\cap N(i)\\ s_j \in \sbf}}
A_{ij} \cdot s_j \Bigg) \geq \max \Bigg( \tbf_i + \sum_{\substack{j \in 
X_{v}\cap N(i)\\ s_j \in \sbf}} A_{ij} \cdot s_j \Bigg) - \frac{\eps}{4}.
\end{equation*}
By the inductive hypothesis, we have that
$\payv$ is within $\tbf_u$ with an additive error of $r_u$. Observe that, since
there are at most $n$ forget nodes in the tree decomposition, we have $r_u  \le n
\cdot \frac{\eps \cdot f(v)}{4n} \le \frac{\eps}{4}$. Hence, we have:
\begin{equation*}
s_i \cdot \Bigg( \payv_i + \sum_{\substack{j \in X_{v}\cap N(i)\\ s_j \in \sbf}}
A_{ij} \cdot s_j \Bigg) \geq \max \Bigg( \payv_i + \sum_{\substack{ \in 
X_{v}\cap N(i)\\ s_j \in \sbf}} A_{ij} \cdot s_j \Bigg) - \frac{\eps}{4} - 2 \cdot
\frac{\eps}{4}.
\end{equation*}
Thus, 
Inequality~\eqref{eqn:forget2} is satisfied and so 
$(\sbf, \payv)$ will not be discarded.

We must still show that the vector $\payv'$ constructed by the algorithm
satisfies the requirements. Again, by the inductive hypothesis we have
that $\payv$ is within $\tbf_u$ with an additive error of $r_u$. We also have
that $(\tbf_v)_j = (\tbf_u)_j + A_{ji} \cdot s_i$ for all players $j$. Our
algorithm does the same operation, but we must account for the extra rounding
step, which can add an additional error of at most $\eps/4n$. Since $r_v = r_u +
\eps/4n$, the resulting candidate satisfies the required properties.

If $v$ is a join node, then let $(\sbf_{u1}, \payv_{u1})$ and
$(\sbf_{u2}, \payv_{u2})$ be the two candidates whose existence is implied by
the inductive hypothesis. We argue that the candidate $(\sbf', \payv')$ produced
by the algorithm from these two candidates will satisfy the requirements. The
first property is obviously true. For the property regarding $\payv'$, observe
that we add a payoff vector with additive error $r_{u1}$ to a payoff vector with
additive error $r_{u2}$, and so the resulting error will be $r_{u1} + r_{u2}$.
Since $f(v) = f(u_1) + f(u_2)$, the required property is satisfied.\qed
\end{proof}

Let $r$ be the root node of $T$. As we explained, we can assume that $X_r$ is
the empty set, and so contains no players. Thus, $C(r)$ is either empty, or it
contains a single witness (a pair of empty sets). The previous lemma implies
that if the game contains a k-uniform $\eps/4$-NE, then $C(r)$ will not be
empty, which will allow us to proceed to phase 2.



\paragraph{\bf Phase 2.}

In phase 2, we assume that $C(r)$ contains a candidate, and we will use this
candidate to find an approximate NE for the game. By Lemma~\ref{lem:soundness}
we know that if the $C(r)$ contains a candidate, then that candidate 
witnesses the existence of an $(\eps + 2r_r)$-NE, where $r_r = \frac{\eps \cdot
f(r)}{4n}$. Since $f(r) \le n$ we have that there exists an $1.5\eps$-NE in the
game, and the goal of phase 2 is to find one.

The algorithm is straightforward. It walks down all branches of the tree
starting at the root. In each step, it has a candidate for the current vertex,
and needs to find corresponding candidates for each child of that node. More
precisely, if $v$
is a node and $(\sbf, \payv)$ is the candidate that has been chosen for that
node (where the sole element of $C(r)$ is chosen for the root initially), then
the algorithm does the following:
\begin{itemize}
\item If $v$ is a join node with two children $u_1$ and $u_2$, then the
algorithm finds the candidates $(\sbf_{u_1}, \payv_{u_1}) \in C(u_1)$ and
candidates $(\sbf_{u_2}, \payv_{u_2}) \in C(u_2)$ that were used to generate
$(\sbf, \payv)$ in phase $1$, and then continues recursively on both $u_1$ and $u_2$.
\item $v$ is a forget node or an introduce node, and $u$ is the child of $v$
then the algorithm finds the candidate $(\sbf', \payv')$ that was used to
generate $(\sbf, \payv)$ and continues recursively on $u$. If $v$ is a forget
node, and player $i$ is forgotten at $v$, then once we have found $(\sbf',
\payv')$, we assign the strategy $\sbf'_i$ to player $i$.
\item If $v$ is a start node, then algorithm stops.
\end{itemize}
Each of these steps can be carried out by re-running the phase 1
algorithm at each node to determine how $(\sbf, \payv)$ was created, so the
running time of phase 2 is no more than the running time of phase 1. 

After the algorithm terminates each player has been assigned a strategy, and we
claim that the resulting strategy profile $s$ is an $1.5\eps$-NE. This follows
from the fact that when we assign strategy $s_i$ to player $i$ in the forget
node $v$, the test from Inequality~\eqref{eqn:forget2} ensures that player $i$
is $\eps$-happy with respect to the payoff vector $\payv$, and the algorithm
subsequently constructs a strategy profile whose payoffs are within $r_v <
\frac{\eps}{2}$ of $\payv$, so player $i$ is $1.5\eps$-happy in the resulting
strategy profile.

\section{Proof of Theorem~\ref{thm:consk}}
\label{app:consk}

With reference to the algorithm defined in Appendix~\ref{app:alg}, we make the
following modifications.
\begin{itemize}
\item At every Start node, the value of $x$ is initialized to
$g(\sbf_\emptyset)$, where $\sbf_\emptyset$ denotes the strategy profile of an
empty game.
\item When we create a witness for an Introduce node $v$ with child $u$, we copy the value for $x$ given in the witness for $u$.
\item When we create a witness for a Forget node $v$ with child $u$, we set the
value of $x$ by computing
$\addf(\tilg(X_u), i, \sbf_i, x)$, where $i$ is the player that is being forgotten, $\sbf_i$
is the strategy for $i$, and $x$ is the value from the witness for $u$.
\item When we create a witness for a Join node $v$ with children $u_1$ and
$u_2$, we set the value of $x$ by computing $\mergef(\tilg(X_{u_1}), \tilg(X_{u_2}), x_1, x_2)$, where
$x_1$ and $x_2$ are the values from the witnesses for $u_1$ and $u_2$,
respectively.
\end{itemize}
Observe that the definition of the functions $\addf$ and $\mergef$ ensure that
$x$ does indeed capture the value that $g$ would assign to the witnessed
strategy profile on $\tilg(X_v)$.

If at any point two witnesses are created that agree on the $\sbf$ and $\payv$
components, but disagree on the $x$ component, then we keep the witness whose
$x$ value is higher, and discard the other one. This ensures that the number of
possible witnesses computed for each tree decomposition node does not increase
relative to the original algorithm, so the running time of the algorithm remains
unchanged (ignoring the extra polynomial factors needed to compute $g$, $\addf$,
and $\mergef$ during the algorithm.)

Phase two of the algorithm is unchanged. We simply select the (unique) witness
$(\sbf, \payv, x) \in C(r)$ and unroll it to obtain a strategy profile $\sbf'$
for $G$. The properties of $\addf$ and $\mergef$ ensure that $g(\sbf') = x$.

We can now proceed to prove Theorem~\ref{thm:consk}.
\begin{proof}
Since the modifications to the algorithm only deal with the extra parameter $x$,
every witness $(\sbf, \payv, x)$ computed by the algorithm corresponds to a
witness $(\sbf, \payv)$ that would have been computed by the original algorithm.
Hence, the fact that the algorithm finds an $1.5\eps$-NE follows immediately
from Theorem~\ref{thm:twmain}.

For the second property, we will use Lemma~\ref{lem:completeness}. This lemma
proves that for every $\eps/4$-NE $\sbf'$ of the game, every tree decomposition
node will have a witness that corresponds to $\sbf'$. The value of $x$ for this
witness will correspond to the value that $g$ would assign to the restriction of
$\sbf'$ to $\tilg(X_v)$. Since the algorithm always discards witnesses with
smaller values of $x$, this ensures that the approximate equilibrium $\sbf$ must
have $g(\sbf) \ge g(\sbf')$. 
\end{proof}

\section{Proof of Theorem~\ref{thm:nouniform}}

\begin{proof}
The proof of this theorem requires a deeper understanding of the techniques 
used to prove Theorem~\ref{thm:small}. They show that, when $k$ is chosen according
to the required bound, one can take any strategy profile $\sbf'$ and randomly sample a
$k$-uniform strategy profile $\sbf$. The crucial property is that, with high
probability, the payoff of each pure strategy under $\sbf'$ are close to
the payoff of the corresponding pure strategy under $\sbf$. In particular,
the value of $k$ used in our algorithm ensures that the payoffs can move by at
most $\eps/16$. 
This, in turn, means that any linear function over the payoffs can change by at
most $c \eps/16$, for some constant $c$, which gives the required inequality.
\qed
\end{proof}

\end{document}